\documentclass[reqno]{amsart}
\usepackage{amsmath,appendix}
\usepackage{amsfonts}
\usepackage{amsmath,amssymb}
\usepackage{graphicx}
\usepackage{color}

\newtheorem{theorem}{Theorem}

\newtheorem{lemma}[theorem]{Lemma}

\newtheorem{definition*}[theorem]{Definition}
\newtheorem*{remark*}{Remark}
\newtheorem{proposition}[theorem]{Proposition}
\newtheorem{remark}[theorem]{Remark}

\newtheorem{example}[theorem]{Example}

\numberwithin{theorem}{section}
\newcommand{\be}%
{\protect\setcounter{equation}{\value{subsubsection}}}  
\newcommand{\ee}%

\makeatletter

\newcommand{\Rmnum}[1]{\expandafter\@slowromancap\romannumeral #1@}
\makeatother

\newcommand{\F}{\ensuremath{\mathbb F}}
\newcommand{\Z}{\ensuremath{\mathbb Z}}
\newcommand{\N}{\ensuremath{\mathbb N}}

\newcommand{\ls}[1]
{\dimen0=\fontdimen6\the\font\lineskip=#1\dimen0
	\advance\lineskip.5\fontdimen5\the\font
	\advance\lineskip-\dimen0
	\lineskiplimit=0.9\lineskip
	\baselineskip=\lineskip
	\advance\baselineskip\dimen0
	\normallineskip\lineskip\normallineskiplimit\lineskiplimit
	\normalbaselineskip\baselineskip
	\ignorespaces}

\begin{document}
	
\bibliographystyle{abbrv}
	
\title{A Class of few-Lee weight $\Z_2[u]$-linear codes using simplicial complexes and minimal codes via Gray map}
	
\author{Pramod Kumar Kewat$^1$ and Nilay Kumar Mondal$^{1,*}$}
	
\footnotetext[1]{
	\small{\,
		Department of Mathematics and Computing, Indian Institute of Technology (Indian School of Mines), Dhanbad 826 004, India.\\}}
	
%\today}
	
\email {pramodk@iitism.ac.in (P.K. Kewat), 
[corresponding author]nilay.17dr000565@am.ism.ac.in (N.K. Mondal)}  
	
\subjclass{94B05, 94A62}
\keywords{Few-Lee weight codes, Mixed alphabet ring, Simplicial complexes, Minimal codes.}

%\maketitle
	
%\footnotetext[0] {The work is supported by ...}
	
\thispagestyle{plain} \setcounter{page}{1}

\begin{abstract}
Recently some mixed alphabet rings are involved in constructing few-Lee weight additive codes with optimal or minimal Gray images using suitable defining sets or down-sets. Inspired by these works, we choose the mixed alphabet ring $\Z_2\Z_2[u]$ to construct a special class of linear code $C_L$ over $\Z_2[u]$ with $u^2=0$ by employing simplicial complexes generated by a single maximal element. We show that $C_L$ has few-Lee weights by determining the Lee weight distribution of $C_L$. Theoretically, this shows that we may employ simplicial complexes to obatin few-weight codes even in the case of mixed alphabet rings. We show that the Gray image of $C_L$ is self-orthogonal and we have an infinite family of minimal codes over $\Z_2$ via Gray map, which can be used to secret sharing schemes.

\end{abstract}

\maketitle
\markboth{P.K. Kewat, N.K. Mondal}{A class of few-Lee weight $\Z_2[u]$-linear codes from simplicial complexes and minimal codes via Gray map}

\section{Introduction}\label{section 1}
Constructing few-weight codes over finite fields is a long standing area of interest in algebraic coding theory. It has key applications in graph theory, finite geometry and combinatorial designs (see \cite{calderbank1986applications}). The authors, in \cite{shi2016first} first introduced few-Lee weight codes over the finite ring $\F_2+u\F_2$ with $u^2=0$ and as an application they found some infinite families of optimal (with respect to the Griesmer bound) few-weight codes over $\F_2$ through Gray map. Apart from this, using a famous characterization of minimal codes, namely the Ashikhmin-Barg lemma, they showed that some of the Gray images are also minimal over $\F_2$, which have direct applications in secret sharing schemes and authentication codes (see \cite{anderson1998secret}, \cite{carlet2005secret}, \cite{ding2007authentication}, \cite{yuan2006secret}). Actually finding minimal codes over finite fields is itself a central topic in algebraic coding theory and in \cite{shi2016first} the authors succesfully clubbed these two important topics together. After that lots of researchers have constructed several classes of few-Lee weight codes over different finite chain rings of characteristic $p$ having optimal or minimal codes as appliactions. In particular for $p=2$, see \cite{shi2018trace2}, \cite{shi2021simplicial}, \cite{shi2021simplicial2}, \cite{shi2021trace3}, \cite{shi2020generator}, \cite{wu2020simplicial}. The researchers mainly used suitable defining sets, simplicial complexes or down-sets (for $p>2$ only) to construct such codes (see \cite{chang2018simplicial}, \cite{ding2015trace}, \cite{wu2020down}). Very recently the authors in \cite{shi2022nonchain} first studied few-Lee weight codes over a non-chain ring $\F_2+u\F_2+v\F_2+uv\F_2$.
\vskip 2pt
The authors, in \cite{dougherty2016mixedfirst} constructed one-Lee weight additive codes over the mixed alphabet ring $\Z_2\Z_4$. In \cite{shi2020mixedsecond}, the authors studied one-Lee weight and two-Lee weight codes over $\Z_2\Z_2[u,v]$ with $u^2=0,v^2=0$ and $uv=vu$. The authors, in \cite{sole2021mixedthird} generalized \cite{dougherty2016mixedfirst} and studied the construction of one-Lee weight and two-Lee weight $\Z_2\Z_4[u]$-additive codes.
In \cite{wu2021mixedtrace}, the authors constructed few-Lee weight additive codes over $\Z_p\Z_p[u]$ with $u^2=0$ having good codes as Gray images by using a suitable defining set $D=(D_0^e,D_0^e+u\F_{p^m})\subseteq\F_{p^m}\times(\F_{p^m}+u\F_{p^m})$, where $D_0^e$ is the cyclotomic class of order e in $\F_{p^m}$. After that the authors, in \cite{wang2021mixeddown}, raised a natural question that whether we can construct such few-Lee weight codes with minimal or optimal Gray images over mixed alphabet rings by employing down-sets or simplicial complexes and they provided a partial answer to that by using the mixed alphabet ring  $\Z_p\Z_p[u]$ with $u^2=u$ and $p>2$ to study some few-Lee weight additive codes by using the set $L=\{t_1,ut_2+(1-u)t_3~|~t_i\in\Delta_i^c\}\subseteq(\Z_{p}^m\times\Z_{p}^m+u\Z_{p}^m)$ consisting of down-sets $\Delta_i\subseteq\Z_{p}^m$ generated by a single element.
\vskip 5pt
In this paper, we choose the mixed alphabet ring $\Z_2\Z_2[u]$ and defined an inner product on $\Z_2^m\Z_2^m[u]$ to construct a class of linear codes over $\Z_2[u]$ with $u^2=0$ using the set $L=\{t_1,t_2+ut_3~|~t_i\in\Delta_i^c\}\subseteq(\Z_{2}^m\times\Z_{2}^m+u\Z_{2}^m)$ consisting of simplicial complexes $\Delta_i\subseteq\Z_{2}^m$ generated by a single element. Further we compute the Lee weight distribution for this class of codes and find it to be of few-Lee weight. To the best of our knowledge, this is the first time simplicial complexes are being used in the case of mixed alphabet rings and we must note that down-sets (a subclass of simplicial complexes) can be used only in the case $p>2$. So it gives the complete answer to the question raised in \cite{wang2021mixeddown}. We have also studied their Gray images and show that the Gray image of $C_L$ is self-orthogonal. Moreover, we have an infinite family of binary minimal codes with the help of Ashikhmin-Barg lemma.
\vskip 2pt
This paper is organized as follows. Section 2 covers the basic definitions and related prerequisites. In Section 3, we compute the Lee weight distributions of the linear code $C_L$ over $\Z_2[u]$ as defined in Section 2. As an application, we also find an infinite family of minimal codes over $\Z_2$ as Gray images of $C_L$ in Section 4. Finally we give an overview of this material in Section 5.

\section{Preliminaries}\label{section 2}
\subsection{Rings and simplicial complexes}
Let $\Z_2$ be the ring of integers modulo $2$ and $\Z_2[u]=\Z_2+u\Z_2,~u^2=0$. We define  $\mathcal{R}=\Z_2\Z_2[u]=\{(x,y+uz)|x,y,z\in\Z_2\}$ and $\mathcal{R}^m=\{(p,q+ur)|p,q,r\in\Z_2^m\}$, where $m$ is any positive integer strictly greater than 1. Since $\Z_2$ is a subring of $\Z_2[u]$, it is easy to observe that $\mathcal{R}^m$ is a $\Z_2[u]$-module under component-wise addition and $\Z_2[u]$-scalar multiplication defined as: $(p_1,q_1+ur_1)+(p_2,q_2+ur_2)=((p_1+p_2),(q_1+q_2)+u(r_1+r_2))$ and $(y+uz)(p,q+ur)=(yp,yq+u(yr+zq))$.
\vskip 2pt
For any vector $w\in\Z_2^m$, the support of $w$ denoted by $supp(w)$ is defined as the set of all nonzero coordinate positions of $w$. Clearly, $w\mapsto supp(w)$ is a bijection between $\Z_2^m$ and the power set of $[m]=\{1,\ldots,m\}$ i.e., $2^{[m]}$. %Throughout this paper we often identify a vector in $\Z_2^m$ with its support. 
Let $\Delta\subseteq\Z_2^m$ such that if $w_1\in\Delta$ and $supp(w_2)\subseteq supp(w_1)$ imply $w_2\in\Delta$, then $\Delta$ is called a simplicial complex. For a vector $w\in\Delta$, if there does not exist any vector $v\in\Delta$ such that $supp(w)\subseteq supp(v)$, then we call $w$ is a maximal element of $\Delta$. Clearly a simplical complex $\Delta$ may admit more than one maximal elements. Here we consider the simplicial complexes with a single maximal element. The simplicial complex $\Delta_{supp(w)}$ generated by a single maximal element $w$ is defined to be the set of all subsets of $supp(w)\subseteq[m]$ and is of size $2^{|supp(w)|}$, where $|\cdot|$ denotes the size of a set.

\subsection{Lee weight, inner product and code}\label{subsec 2.2}
A Gray map $\phi:\Z_2[u]\rightarrow\Z_2^2$ is defined as $y+uz\mapsto(z,y+z)$, where $y,z\in\Z_2$. %Now we can extend the above definrd Gray map $\phi'$ to $\phi:\mathcal{R}=\Z_2\Z_2[u]\rightarrow\Z_2^3$ given by $\phi(x,y+uz)=(x,z,y+z),\forall x,y,z\in\Z_2$. Also 
We can naturally extend the map $\phi$ to $\Phi:\Z_2^m+u\Z_2^m\rightarrow\Z_2^{2m}$ as %Let $a=(p,q+ur)$, where $p,q,r\in\Z_2^m$. 
$\Phi(q+ur)=(r,q+r)$, where $q,r\in\Z_2^m$. For a vector $w\in\Z_2^m$, the Hamming weight of $w$ denoted by $wt_H(w)$ is defined as $wt_H(w)=|supp(w)|$. The Lee weight of a vector $w'=(q+ur)\in\Z_2^m+u\Z_2^m$ denoted by $wt_L(a)$ is defined as its Hamming weight under the Gray map $\Phi$, i.e. $wt_L(w')=wt_L(q+ur)=wt_H(r)+wt_H(q+r)$. The Lee distance of $w'_1,w'_2\in\Z_2^m+u\Z_2^m$ is defined as $wt_L(w'_1-w'_2)$. It is quite easy to check that $\Phi$ is an isometry between $(\Z_2^m+u\Z_2^m,d_L)$ and $(\Z_2^{2m},d_H)$. 
\vskip 2pt
A linear code $C$ of length $m$ over a finite commutative ring $R$ is an $R$-submodule of $R^m$. The Euclidean inner product of vectors $a=(a_1,a_2,\ldots,a_m),b=(b_1,b_2,\ldots,b_m)\in R^m$ is $\langle a,b\rangle=\sum\limits_{i=1}^ma_ib_i\in R$. Let $w_1,w_2\in\Z_2^m$ and $w'_1,w'_2\in\Z_2^m+u\Z_2^m$. Now we define the inner product of two vectors $(w_1,w'_1)$ and $(w_2,w'_2)$ in $\mathcal{R}^m$ as $((w_1,w'_1)\cdot(w_2,w'_2))=u\langle w_1,w_2\rangle+\langle w'_1,w'_2\rangle$.
\vskip 2pt
Let $\Delta_1,\Delta_2,\Delta_3\subseteq\Z_2^m$ be three simplicial complexes generated by a single maximal elements having support $D,E,F\subseteq[m]$ respectively and they are not equal to $\Z_2^m$. From now on we denote $\Delta_1,\Delta_2,\Delta_3$ as $\Delta_D,\Delta_E,\Delta_F$ throughout this paper. We set $L=\{l=(t_1,t_2+ut_3)|t_1\in\Delta_D^c,t_2\in\Delta_E^c,t_3\in\Delta_F^c\}\subseteq\mathcal{R}^m$, where $\Delta^c=\Z_2^m\setminus\Delta$. We define $C_L=\{c_a=((a\cdot l))_{l\in L}|a\in\mathcal{R}^m\}$. It is easy to check that $C_L$ is a code over $\Z_2[u]$ of length $|L|$. The code $C_L$ is said to be a $t$-Lee weight code if all of its codewords have only $t$ different nonzero Lee weights.

\subsection{Generating function and characteristic function}
Let $X$ be a subset of $\Z_2^m$. Chang and Hyun (\cite{chang2018simplicial}) introduced the following $m$-varriable generating function associated with the set $X$:
\begin{equation}\label{eq 1}
	\mathcal{H}_X(x_1,x_2,\ldots,x_m)=\sum\limits_{u\in X}\prod\limits_{i=1}^m(x_i)^{u_i}\in\Z[x_1,x_2,\ldots,x_m], ~\text{where}~u=(u_1,u_2,\ldots,u_m)\in\Z_2^m.
\end{equation}
\begin{lemma}\cite[Lemma 2]{chang2018simplicial}\label{lem 2.1} Let $\Delta\subseteq\Z_2^m$ be a simplicial complex. Then the following results hold.
\begin{itemize}
	%\item [(1)] $\mathcal{H}_\emptyset(x_1,x_2,\ldots,x_m)=0$,
	%\item [(2)] $\mathcal{H}_{\Z_2^m}(x_1,x_2,\ldots,x_m)=\prod\limits_{i\in[m]}(1+x_i)$,
	\item [(1)] $\mathcal{H}_\Delta(x_1,x_2,\ldots,x_m)+\mathcal{H}_{\Delta^c}(x_1,x_2,\ldots,x_m)=\mathcal{H}_{\Z_2^m}(x_1,x_2,\ldots,x_m)=\prod\limits_{i\in[m]}(1+x_i)$,
	\item [(2)] $\mathcal{H}_{\Delta_S}(x_1,x_2,\ldots,x_m)=\prod\limits_{i\in S}(1+x_i)$, where $\Delta_S$ is a simplicial complex generated by a single maximal elemnt having support $S\subseteq[m]$.
\end{itemize}
\end{lemma}

\vskip 2pt
Suppose $X$ and $Y$ are two subsets of $[m]$. We define a characteristic function $\chi:2^{[m]}\times2^{[m]}\rightarrow\{0,1\}$ as $\chi(X|Y)=1$ if and only if $X\cap Y=\emptyset$ and zero otherwise.
\begin{lemma}\cite[Lemma 3.1]{wu2020simplicial}\label{lem 2.2} Let $D,E\subseteq[m]$. %as mentioned in \ref{subsec 2.2}. 
Then
	\begin{itemize}
		\item [(1)] $|\{X\subseteq[m]~|~\emptyset\ne X,\chi(X|D)=1\}|=2^{m-|D|}-1$ and $|\{X\subseteq[m]~|~\chi(X|D)=0\}|=2^m-2^{m-|D|}$.
		\item [(2)] $|\{X\subseteq[m]~|~\emptyset\ne X,\chi(X|D)\chi(X|E)=1\}|=2^{m-|D\cup E|}-1$ and $|\{X\subseteq[m]~|~\chi(X|D)\chi(X|E)=0\}|=2^m-2^{m-|D\cup E|}$.
		\item [(3)] Define $D\oplus E=(D\cup E)\setminus(D\cap E)$. Let $T_i=|\{(X,Y)~|~\emptyset\ne X,Y\subseteq[m],X\ne Y,\chi(Y|E)=1,\chi(X|D)+\chi(X\oplus Y|D)=i\}|$, where $i\in\{0,1,2\}$. Then $T_0=2^m(2^{m-|E|}-1)+2^{m-|D|}(1+2^{m-|D\cup E|}-2^{m+1-|E|})$, $T_1=2(2^{m-|D|}-1)(2^{m-|E|}-2^{m-|D\cup E|})$, and $T_2=(2^{m-|D|}-2)(2^{m-|D\cup E|}-1)$.
	\end{itemize} 
\end{lemma}

\section{Main Result}\label{section 3}

%\subsection{Lee weight distribution of $C_L$}
Let $a=(p,q+ur),l=(t_1,t_2+ut_3)$, where $a\in\mathcal{R}^m$ and $l\in L$. The Lee weigt of the codeword $c_a$ of $C_L$ is
\begin{align*}
	wt_L(c_a)&=wt_L(((p,q+ur)\cdot(t_1,t_2+ut_3))_{t_1,t_2,t_3})\\
	&=wt_L((u\langle p,t_1\rangle+\langle q+ur,t_2+ut_3\rangle)_{t_1,t_2,t_3})\\
	&=wt_L((q t_2+u(p t_1+q t_3+r t_2))_{t_1,t_2,t_3})\\
	&=wt_H((p t_1+q t_3+r t_2)_{t_1,t_2,t_3})+wt_H((p t_1+(q+r) t_2+q t_3)_{t_1,t_2,t_3})\\
	&=|L|-\frac{1}{2}\sum\limits_{y\in\Z_2}\sum\limits_{t_1\in\Delta_D^c}\sum\limits_{t_2\in\Delta_E^c}\sum\limits_{t_3\in\Delta_F^c}(-1)^{(p t_1+q t_3+r t_2)y}\\
	&~+|L|-\frac{1}{2}\sum\limits_{y\in\Z_2}\sum\limits_{t_1\in\Delta_D^c}\sum\limits_{t_2\in\Delta_E^c}\sum\limits_{t_3\in\Delta_F^c}(-1)^{(p t_1+(q+r) t_2+q t_3)y}\\
	&=|L|-\frac{1}{2}\sum\limits_{t_1\in\Delta_D^c}(-1)^{p t_1}\sum\limits_{t_2\in\Delta_E^c}(-1)^{r t_2}\sum\limits_{t_3\in\Delta_F^c}(-1)^{q t_3}\\
	&~-\frac{1}{2}\sum\limits_{t_1\in\Delta_D^c}(-1)^{p t_1}\sum\limits_{t_2\in\Delta_E^c}(-1)^{(q+r) t_2}\sum\limits_{t_3\in\Delta_F^c}(-1)^{q t_3}.
\end{align*}
Using Equation \eqref{eq 1}, it is easy to observe that for $X=\Delta_D$ and $x_i=(-1)^{p_i}$, where $\Delta_D$ is defined as above and $p=(p_1,p_2,\ldots,p_m)\in\Z_2^m$, we have $\mathcal{H}_{\Delta_D}((-1)^{p_1},(-1)^{p_2},\ldots,(-1)^{p_m})=\sum\limits_{t\in \Delta_D}\prod\limits_{i=1}^m(-1)^{p_it_i}=\sum\limits_{t\in\Delta_D}(-1)^{p t}$, where $t=(t_1,t_2,\ldots,t_m)\in\Delta_D\subseteq\Z_2^m$. So, by Lemma \ref{lem 2.1}, $\sum\limits_{t\in\Delta_D^c}(-1)^{p t}= \mathcal{H}_{\Delta_D^c}((-1)^{p_1},(-1)^{p_2},\ldots,(-1)^{p_m})=\prod\limits_{i\in[m]}(1+(-1)^{p_i})-\mathcal{H}_{\Delta_D}((-1)^{p_1},(-1)^{p_2},\ldots,(-1)^{p_m})$ $=2^m\delta_{0,p}-\mathcal{H}_{\Delta_D}((-1)^{p_1},(-1)^{p_2},\ldots,(-1)^{p_m})$, where $\delta_{i,j}$ is the Kronecker delta function i.e. $\delta_{i,j}=1$ if and only if $i=j$ and zero otherwise. Also we have $\mathcal{H}_{\Delta_D}((-1)^{p_1},(-1)^{p_2},\ldots,(-1)^{p_m})=\prod\limits_{i\in D}(1+(-1)^{p_i})=2^{|D|}\chi(supp(p)|D)$, where $\chi$ is the characteristic function on the set $2^{[m]}\times2^{[m]}$ defined in Section \ref{section 2}. 
\vskip 2pt
So, further we can write
\begin{align}\label{eq 2}
	wt_L(c_a)=\nonumber&|L|-\frac{1}{2}[(2^m\delta_{0,p}-2^{|D|}\chi(supp(p)|D))(2^m\delta_{0,r}-2^{|E|}\chi(supp(r)|E))\\\nonumber&(2^m\delta_{0,q}-2^{|F|}\chi(supp(q)|F))]-\frac{1}{2}[(2^m\delta_{0,p}-2^{|D|}\chi(supp(p)|D))\\
	&(2^m\delta_{0,(q+r)}-2^{|E|}\chi(supp(q+r)|E))(2^m\delta_{0,q}-2^{|F|}\chi(supp(q)|F))].
\end{align}

\begin{lemma}\label{lem 3.1}
	Let $D,E,F\subseteq[m]$. %as mentioned in \ref{subsec 2.2}.
Then
	\begin{itemize}
		\item [(1)] $|\{X\subseteq[m]~|~\chi(X|D)=0,\chi(X|E)=0\}|=2^m-(2^{m-|D\cup E|})[(2^{|D|-|D\cap E|})+(2^{|E|-|D\cap E|})-1]$ and $|\{X\subseteq[m]~|~\chi(X|D)=0,\chi(X|E)=1\}|=(2^{m-|D\cup E|})(2^{|D|-|D\cap E|}-1)$ and $|\{X\subseteq[m]~|~\chi(X|D)=1,\chi(X|E)=0\}|=(2^{m-|D\cup E|})(2^{|E|-|D\cap E|}-1)$.
		\item [(2)] Let $T'_i=|\{(X,Y)~|~\emptyset\ne X,Y\subseteq[m],X\ne Y,\chi(Y|E)=0,\chi(X|D)+\chi(X\oplus Y|D)=i\}|$, where $i\in\{0,1,2\}$. Then $\sum\limits_{i=1}^{i=3}T'_i=(2^m-2^{m-|E|})(2^m-2)$.
%		\item [(3)] $|\{X\subseteq[m]~|~\emptyset\ne X,\chi(X|D)\chi(X|E)\chi(X|F)=1\}|=2^{m-|D\cup E\cup F|}-1$ and $|\{X\subseteq[m]~|~,\chi(X|D)\chi(X|E)\chi(X|F)=0\}|=2^m-(2^{m-|D\cup E\cup F|})$.
%		\item [(4)] 
%		$S_{F_1}=|\{X\subseteq[m]~|~\chi(X|D)=0,\chi(X|E)=0,\chi(X|F)=1\}|=(2^{m-|D\cup E\cup F|})(2^{|D\cap E|-|D\cap E\cap F|})(2^{|D|-|D\cap (E\cup F)|}-1)(2^{|E|-|E\cap (D\cup F)|}-1)$ and $S_{E_1}=|\{X\subseteq[m]~|~\chi(X|D)=0,\chi(X|E)=1,\chi(X|F)=0\}|=(2^{m-|D\cup E\cup F|})(2^{|D\cap F|-|D\cap E\cap F|})(2^{|D|-|D\cap (E\cup F)|}-1)(2^{|F|-|F\cap (D\cup E)|}-1)$ and $S_{D_1}=|\{X\subseteq[m]~|~\chi(X|D)=1,\chi(X|E)=0,\chi(X|F)=0\}|=(2^{m-|D\cup E\cup F|})(2^{|E\cap F|-|D\cap E\cap F|})(2^{|E|-|E\cap (D\cup F)|}-1)(2^{|F|-|F\cap (D\cup E)|}-1)$.
%		\item [(5)]
%		$S_{D_0}=|\{X\subseteq[m]~|~\chi(X|D)=0,\chi(X|E)=1,\chi(X|F)=1\}|=(2^{m-|D\cup E\cup F|})(2^{|D|-|D\cap(E\cup F)|}-1)$ and $S_{E_0}=|\{X\subseteq[m]~|~\chi(X|D)=1,\chi(X|E)=0,\chi(X|F)=1\}|=(2^{m-|D\cup E\cup F|})(2^{|E|-|E\cap(D\cup F)|}-1)$ and $S_{F_0}=|\{X\subseteq[m]~|~\chi(X|D)=1,\chi(X|E)=1,\chi(X|F)=0\}|=(2^{m-|D\cup E\cup F|})(2^{|F|-|F\cap(D\cup E)|}-1)$.
%		\item [(6)]
%		$S_0=|\{X\subseteq[m]~|~\chi(X|D)=0,\chi(X|E)=0,\chi(X|F)=0\}|=2^m-(2^{m-|D\cup E\cup F|})-(S_{D_1}+S_{E_1}+S_{F_1}+S_{D_0}+S_{E_0}+S_{F_0})$.
	
	\end{itemize}
\end{lemma}

\begin{proof}
	\begin{itemize}
		\item [(1)] Let us consider the set $\{X\subseteq[m]~|~\chi(X|D)=0,\chi(X|E)=1\}$. Since $X\cap D\ne\emptyset$, $X$ has $(2^{|D|}-1)$ choices but also $X\cap E=\emptyset$, so the choices redeuces to $(2^{|D|-|D\cap E|}-1)$ and for each of these choices we have $(2^{m-|D\cup E|})$ possibilities as $X\subseteq[m]$. So all total we have $|\{X\subseteq[m]~|~\chi(X|D)=0,\chi(X|E)=1\}|=(2^{m-|D\cup E|})(2^{|D|-|D\cap E|}-1)$. Using similar arguments we have $|\{X\subseteq[m]~|~\chi(X|D)=1,\chi(X|E)=0\}|=(2^{m-|D\cup E|})(2^{|E|-|D\cap E|}-1)$.
		\vskip 2pt
		Now it is easy to observe that $\{X\subseteq[m]~|~\chi(X|D)\chi(X|E)=0\}=\{X\subseteq[m]~|~\chi(X|D)=0,\chi(X|E)=0\}\sqcup \{X\subseteq[m]~|~\chi(X|D)=0,\chi(X|E)=1\}\sqcup \{X\subseteq[m]~|~\chi(X|D)=1,\chi(X|E)=0\}$. So using the above two results and (2) of lemma \ref{lem 2.2}, we have $|\{X\subseteq[m]~|~\chi(X|D)=0,\chi(X|E)=0\}|=2^m-(2^{m-|D\cup E|})[(2^{|D|-|D\cap E|})+(2^{|E|-|D\cap E|})-1]$.
		\item [(2)] It is easy to observe that  $\sum\limits_{i=1}^{i=3}T'_i=|\{(X,Y)~|~ \emptyset\ne X,Y\subseteq[m],X\ne Y,\chi(Y|E)=0\}|$. Hence the result is clear as $Y$ has $(2^m-2^{m-|E|})$ choices (using (1) of Lemma \ref{lem 2.2}) and $X$ has $(2^m-2)$ as $X\ne\emptyset$ and $X\ne Y$.
	\end{itemize}
	
\end{proof}

\begin{theorem}\label{thm 3.2}
	Let $\Delta_D,\Delta_E,\Delta_F\subseteq\Z_2^m$ are three simplicial complexes with single maximal elements having support $D,E,F\subseteq[m]$ respectively. Let $L=\{(t_1,t_2+ut_3)|t_1\in\Delta_D^c,t_2\in\Delta_E^c,t_3\in\Delta_F^c\}$, where $\Delta^c=\Z_2^m\setminus\Delta$. Then the code $C_L$ of length $(2^m-2^{|D|})(2^m-2^{|E|})(2^m-2^{|F|})$ (as defined in Subsection \ref{subsec 2.2}) is linear over $\Z_2[u]$. The Lee weight distribution of $C_L$ is given in Table \ref{Tab 3.1} and the size of $C_L$ is given by
	%\begin{align*}
	$|C_L|=\begin{cases}
	2^{3m-1}&if~|D|=|E|=m-1,\\
	2^{3m}&otherwise.
	\end{cases}$
	%\end{align*}
\end{theorem}

\begin{proof}
	The length of $C_L$ is quite straightforward. To show $C_L$ is linear over $\Z_2[u]$, let us consider the the map $f$ between two $\Z_2[u]$-module $\mathcal{R}^m$ and $\Z_2[u]^{|L|}$ defined as
	\begin{align*}
	f:&\mathcal{R}^m\rightarrow\Z_2[u]^{|L|}\\
	&a\mapsto ((a\cdot l))_{l\in L}
	\end{align*}
	Clearly, $Im(f)=C_L$ and to prove $C_L$ is $\Z_2[u]$-linear, it is enough to show that $f$ is a $\Z_2[u]$-module homomorphism. Let $a_1=(p_1,q_1+ur_1)$ and $a_2=(p_2,q_2+ur_2)$. Then $((a_1\cdot l))=(q_1 t_2+u(p_1 t_1+q_1 t_3+r_1 t_2))$ and $((a_2\cdot l))=(q_2 t_2+u(p_2 t_1+q_2 t_3+r_2 t_2))$ and $(((a_1+a_2)\cdot l))=((q_1+q_2) t_2+u((p_1+p_2) t_1+(q_1+q_2) t_3+(r_1+r_2) t_2))=((a_1\cdot l))+((a_2\cdot l))$ (as $\Z_2[u]$ is a ring under component-wise addition). So, $(((a_1+a_2)\cdot l))_{l\in L}=((a_1\cdot l))_{l\in L}+((a_2\cdot l))_{l\in L}$, i.e., $f(a_1+a_2)=f(a_1)+f(a_2)$. 
	\vskip 2pt
	Again, let $\alpha=(y+uz)\in\Z_2[u]$. Then $((\alpha a_1\cdot l))=(((yp_1,yq_1+u(yr_1+zq_1))\cdot l))=u\langle yp_1,t_1\rangle+\langle yq_1+u(yr_1+zq_1),t_2+ut_3\rangle=\langle yq_1,t_2\rangle+u(\langle yq_1,t_3\rangle+\langle yr_1,t_2\rangle+\langle zq_1,t_2\rangle+\langle yp_1,t_1\rangle)=y\langle q_1,t_2\rangle+u(y\langle q_1,t_3\rangle+y\langle r_1,t_2\rangle+z\langle q_1,t_2\rangle+y\langle p_1,t_1\rangle)=(y+uz)(\langle q_1,t_2\rangle+u(\langle q_1,t_3\rangle+\langle r_1,t_2\rangle+\langle p_1,t_1\rangle))=\alpha ((a_1\cdot l))$. So, $((\alpha a_1\cdot l))_{l\in L}=\alpha ((a_1\cdot l))_{l\in L}$ i.e, $f(\alpha a_1)=\alpha f(a_1)$.
	\vskip 2pt
	So, $C_L$ is linear over $\Z_2[u]$ and $C_L\cong \frac{\mathcal{R}^m}{Ker(f)}$ with $|C_L|=\frac{|\mathcal{R}^m|}{|Ker(f)|}=\frac{2^{3m}}{|Ker(f)|}$. Now, first we compute the Lee weight distribution of $C_L$ and then we compute the size of $C_L$. 
	\vskip 5pt
	To compute the Lee weight distribution of $C_L$, we have $supp(q+r)=supp(q)\oplus supp(r)$ as decsribed in Lemma \ref{lem 2.2}. Let $f_i~1\le i\le57$ be the frequency of $wt_L(c_a)$. Next we divide the proof in the following cases and calculate $f_i$ using Equation \eqref{eq 2}, Lemma \ref{lem 2.2} and Lemma \ref{lem 3.1}.
	\begin{itemize}
		\item [(1)] $p=q=r=0$, $wt_L(c_a)=0$ and $f_1=1$.
		\item [(2)] $p=q=0,r\ne0$, $wt_L(c_a)=|L|-[(2^m-2^{|D|})(-2^{|E|}\chi(supp(r)|E))(2^m-2^{|F|})]$.
		\begin{itemize}
			\item If $\chi(supp(r)|E)=0$, then $wt_L(c_a)=|L|$ and $f_2=2^m-2^{m-|E|}$.
			\item If $\chi(supp(r)|E)=1$, then $wt_L(c_a)=|L|+[(2^m-2^{|D|})(2^{|E|})(2^m-2^{|F|})]$ and $f_3=2^{m-|E|}-1$.
		\end{itemize}
		\item [(3)] $p=0,q\ne0,r=0$, $wt_L(c_a)=|L|-\frac{1}{2}[(2^m-2^{|D|})(2^m-2^{|E|})(-2^{|F|}\chi(supp(q)|F))]-\frac{1}{2}[(2^m-2^{|D|})(-2^{|E|}\chi(supp(q)|E))(-2^{|F|}\chi(supp(q)|F))]$.
		\begin{itemize}
			\item If $\chi(supp(q)|E)=0$,  $\chi(supp(q)|F)=0$, then $wt_L(c_a)=|L|$ and $f_4=2^m-(2^{m-|E\cup F|})$ $[(2^{|E|-|E\cap F|})+(2^{|F|-|E\cap F|})-1]$.
			\item If $\chi(supp(q)|E)=1$, $\chi(supp(q)|F)=0$, then $wt_L(c_a)=|L|$ and $f_5=(2^{m-|E\cup F|})$ $(2^{|F|-|E\cap F|}-1)$.
			\item If $\chi(supp(q)|E)=0$,  $\chi(supp(q)|F)=1$, then $wt_L(c_a)=|L|+\frac{1}{2}[(2^m-2^{|D|})(2^m-2^{|E|})(2^{|F|})]$ and $f_6=(2^{m-|E\cup F|})(2^{|E|-|E\cap F|}-1)$.
			\item If $\chi(supp(q)|E)=1$, $\chi(supp(q)|F)=1$, $wt_L(c_a)=|L|+\frac{1}{2}[(2^m-2^{|D|})(2^m-2^{|E|})(2^{|F|})]-\frac{1}{2}[(2^m-2^{|D|})(2^{|E|})(2^{|F|})]$ and $f_7=2^{m-|E\cup F|}-1$.
		\end{itemize}
		\item [(4)] $p=0,q\ne0,r\ne0$, 	$wt_L(c_a)=|L|-\frac{1}{2}[(2^m-2^{|D|})(-2^{|E|}\chi(supp(r)|E))(-2^{|F|}\chi(supp(q)|F))]-\frac{1}{2}[(2^m-2^{|D|})(2^m\delta_{0,(q+r)}-2^{|E|}\chi(supp(q+r)|E))(-2^{|F|}\chi(supp(q)|F))]$.
		\begin{itemize}
			\item Let $q\ne r$, $wt_L(c_a)=|L|-\frac{1}{2}[(2^m-2^{|D|})(-2^{|E|}\chi(supp(r)|E))(-2^{|F|}\chi(supp(q)|F))]-\frac{1}{2}[(2^m-2^{|D|})(-2^{|E|}\chi(supp(q+r)|E))(-2^{|F|}\chi(supp(q)|F))]$.
			\begin{itemize}
				\item If $\chi(supp(q)|F)=0$, $\chi(supp(r)|E)=0$, $\chi(supp(q+r)|E)=0$, then $wt_L(c_a)=|L|$ and let $f_8$ be the frequency in this case.
				\item If $\chi(supp(q)|F)=0$, $\chi(supp(r)|E)=0$, $\chi(supp(q+r)|E)=1$, then $wt_L(c_a)=|L|$ and let $f_9$ be the frequency in this case.
				\item If $\chi(supp(q)|F)=0$, $\chi(supp(r)|E)=1$, $\chi(supp(q+r)|E)=0$, then $wt_L(c_a)=|L|$ and let $f_{10}$ be the frequency in this case.
				\item If $\chi(supp(q)|F)=0$, $\chi(supp(r)|E)=1$, $\chi(supp(q+r)|E)=1$, then $wt_L(c_a)=|L|$ and let $f_{11}$ be the frequency in this case. Now, $\sum\limits_{i=8}^{11}f_i=(2^m-2^{m-|F|})(2^m-2)$.
				\item If $\chi(supp(q)|F)=1$, $\chi(supp(r)|E)=0$, $\chi(supp(q+r)|E)=0$, then $wt_L(c_a)=|L|$ and $f_{12}=2^m(2^{m-|F|}-1)+2^{m-|E|}(1+2^{m-|E\cup F|}-2^{m+1-|F|})$.
				\item If $\chi(supp(q)|F)=1$, $\chi(supp(r)|E)=0$, $\chi(supp(q+r)|E)=1$, then $wt_L(c_a)=|L|-\frac{1}{2}[(2^m-2^{|D|})(2^{|E|})(2^{|F|})]$ and let $f_{13}$ be the frequency in this case.
				\item If $\chi(supp(q)|F)=1$, $\chi(supp(r)|E)=1$, $\chi(supp(q+r)|E)=0$, then $wt_L(c_a)=|L|-\frac{1}{2}[(2^m-2^{|D|})(2^{|E|})(2^{|F|})]$ and let $f_{14}$ be the frequency in this case. Now, $f_{13}+f_{14}=2(2^{m-|E|}-1)(2^{m-|F|}-2^{m-|E\cup F|})$.
				\item If $\chi(supp(q)|F)=1$, $\chi(supp(r)|E)=1$, $\chi(supp(q+r)|E)=1$, then $wt_L(c_a)=|L|-\frac{1}{2}[(2^m-2^{|D|})(2^{|E|})(2^{|F|})]-\frac{1}{2}[(2^m-2^{|D|})(2^{|E|})(2^{|F|})]$ and $f_{15}=(2^{m-|E|}-2)(2^{m-|E\cup F|}-1)$.
			\end{itemize}
			\item Let $q=r$, $wt_L(c_a)=|L|-\frac{1}{2}[(2^m-2^{|D|})(-2^{|E|}\chi(supp(r)|E))(-2^{|F|}\chi(supp(q)|F))]-\frac{1}{2}[(2^m-2^{|D|})(2^m-2^{|E|})(-2^{|F|}\chi(supp(q)|F))]$.
				\begin{itemize}
				\item If $\chi(supp(r)|E)=0$,  $\chi(supp(q)|F)=0$, then $wt_L(c_a)=|L|$ and $f_{16}=2^m-(2^{m-|E\cup F|})[(2^{|E|-|E\cap F|})+(2^{|F|-|E\cap F|})-1]$.
				\item If $\chi(supp(r)|E)=1$, $\chi(supp(q)|F)=0$, then $wt_L(c_a)=|L|$ and $f_{17}=(2^{m-|E\cup F|})$ $(2^{|F|-|E\cap F|}-1)$.
				\item If $\chi(supp(r)|E)=0$,  $\chi(supp(q)|F)=1$, then $wt_L(c_a)=|L|+\frac{1}{2}[(2^m-2^{|D|})(2^m-2^{|E|})(2^{|F|})]$ and $f_{18}=(2^{m-|E\cup F|})(2^{|E|-|E\cap F|}-1)$.
				\item If $\chi(supp(r)|E)=1$, $\chi(supp(q)|F)=1$, $wt_L(c_a)=|L|-\frac{1}{2}[(2^m-2^{|D|})(2^{|E|})(2^{|F|})]+\frac{1}{2}[(2^m-2^{|D|})(2^m-2^{|E|})(2^{|F|})]$ and $f_{19}=2^{m-|E\cup F|}-1$.
			\end{itemize}
		\end{itemize}
		\item [(5)] $p\ne0,q=r=0$, $wt_L(c_a)=|L|+[(2^{|D|}\chi(supp(p)|D))(2^m-2^{|E|})(2^m-2^{|F|})]$.
		\begin{itemize}
			\item If $\chi(supp(p)|D)=0$, then $wt_L(c_a)=|L|$ and $f_{20}=2^m-2^{m-|D|}$.
			\item If $\chi(supp(p)|D)=1$, then $wt_L(c_a)=|L|+[(2^{|D|})(2^m-2^{|E|})(2^m-2^{|F|})]$ and $f_{21}=2^{m-|D|}-1$.
		\end{itemize}
		\item [(6)] $p\ne0,q=0,r\ne0$, $wt_L(c_a)=|L|-\frac{1}{2}[(-2^{|D|}\chi(supp(p)|D))(-2^{|E|}\chi(supp(r)|E))(2^m-2^{|F|})]-\frac{1}{2}[(-2^{|D|}\chi(supp(p)|D))(-2^{|E|}\chi(supp(r)|E))(2^m-2^{|F|})]$.
		%\begin{itemize}
		%\item Let $p\ne r$.
			\begin{itemize}
			\item If $\chi(supp(p)|D)=0$,  $\chi(supp(r)|E)=0$, then $wt_L(c_a)=|L|$ and $f_{22}=(2^m-2^{m-|D|})(2^m-2^{m-|E|})$.
			\item If If $\chi(supp(p)|D)=1$,  $\chi(supp(r)|E)=0$, then $wt_L(c_a)=|L|$ and $f_{23}=(2^{m-|D|}-1)(2^m-2^{m-|E|})$.
			\item If $\chi(supp(p)|D)=0$,  $\chi(supp(r)|E)=1$, then $wt_L(c_a)=|L|$ and $f_{24}=(2^m-2^{m-|D|})$ $(2^{m-|E|}-1)$.
			\item If $\chi(supp(p)|D)=1$,  $\chi(supp(r)|E)=1$, $wt_L(c_a)=|L|-[(2^{|D|})(2^{|E|})(2^m-2^{|F|})]$ and $f_{25}=(2^{m-|D|}-1)(2^{m-|E|}-1)$.
			\end{itemize}
%		\item Let $p=r$.
%			\begin{itemize}
%				\item If $\chi(supp(p)|D)=0$,  $\chi(supp(r)|E)=0$, then $wt_L(c_a)=|L|$ and $f'_{22}=2^m-(2^{m-|D\cup E|})[(2^{|D|-|D\cap E|})+(2^{|F|-|D\cap E|})-1]$.
%				\item If If $\chi(supp(p)|D)=1$,  $\chi(supp(r)|E)=0$, then $wt_L(c_a)=|L|$ and $f'_{23}=(2^{m-|D\cup E|})(2^{|E|-|D\cap E|}-1)$.
%				\item If $\chi(supp(p)|D)=0$,  $\chi(supp(r)|E)=1$, then $wt_L(c_a)=|L|$ and $f'_{24}=(2^{m-|D\cup E|})(2^{|D|-|D\cap E|}-1)$.
%				\item If $\chi(supp(p)|D)=1$,  $\chi(supp(r)|E)=1$, $wt_L(c_a)=|L|-[(2^{|D|})(2^{|E|})(2^m-2^{|F|})]$ and $f'_{25}=2^{m-|D\cup E|}-1$.
%			\end{itemize}
		%\end{itemize}
		\item [(7)] $p\ne0,q\ne0,r=0$, $wt_L(c_a)=|L|-\frac{1}{2}[(-2^{|D|}\chi(supp(p)|D))(2^m-2^{|E|})(-2^{|F|}\chi(supp(q)|F))]-\frac{1}{2}[(-2^{|D|}\chi(supp(p)|D))(-2^{|E|}\chi(supp(q)|E))(-2^{|F|}\chi(supp(q)|F))]$.
%		\begin{itemize}
%			\item Let $p\ne q$.
		\begin{itemize}
			\item If $\chi(supp(p)|D)=0$,  $\chi(supp(q)|F)=0$, $\chi(supp(q)|E)=0$, then $wt_L(c_a)=|L|$ and let $f_{26}$ be the frequency in this case.
			\item If $\chi(supp(p)|D)=0$,  $\chi(supp(q)|F)=0$, $\chi(supp(q)|E)=1$, then $wt_L(c_a)=|L|$ and let $f_{27}$ be the frequency in this case.
			\item If $\chi(supp(p)|D)=0$,  $\chi(supp(q)|F)=1$, $\chi(supp(q)|E)=0$, then $wt_L(c_a)=|L|$ and let $f_{28}$ be the frequency in this case.
			\item If $\chi(supp(p)|D)=0$,  $\chi(supp(q)|F)=1$, $\chi(supp(q)|E)=1$, then $wt_L(c_a)=|L|$ and let $f_{29}$ be the frequency in this case. Now, $\sum\limits_{i=26}^{29}f_i=(2^m-2^{m-|D|})(2^m-1)$.
			\item If $\chi(supp(p)|D)=1$,  $\chi(supp(q)|F)=0$, $\chi(supp(q)|E)=0$, then $wt_L(c_a)=|L|$ and $f_{30}=(2^{m-|D|}-1)[2^m-(2^{m-|E\cup F|})\{(2^{|E|-|E\cap F|})+(2^{|F|-|E\cap F|})-1\}]$.
			\item If $\chi(supp(p)|D)=1$,  $\chi(supp(q)|F)=0$, $\chi(supp(q)|E)=1$, then $wt_L(c_a)=|L|$ and $f_{31}=(2^{m-|D|}-1)[(2^{m-|E\cup F|})(2^{|F|-|E\cap F|}-1)]$.
			\item If $\chi(supp(p)|D)=1$,  $\chi(supp(q)|F)=1$, $\chi(supp(q)|E)=0$, then $wt_L(c_a)=|L|-\frac{1}{2}[(2^{|D|})(2^m-2^{|E|})(2^{|F|})]$ and $f_{32}=(2^{m-|D|}-1)[(2^{m-|E\cup F|})(2^{|E|-|E\cap F|}-1)]$.
			\item If $\chi(supp(p)|D)=1$,  $\chi(supp(q)|F)=1$, $\chi(supp(q)|E)=1$, then $wt_L(c_a)=|L|-\frac{1}{2}[(2^{|D|})(2^m-2^{|E|})(2^{|F|})]+\frac{1}{2}[(2^{|D|})(2^{|E|})(2^{|F|})]$ and $f_{33}=(2^{m-|D|}-1)[2^{m-|E\cup F|}-1]$.
		\end{itemize}
%			\item Let $p=q$.
%			\begin{itemize}
%				\item If $\chi(supp(p)|D)=0$,  $\chi(supp(q)|F)=0$, $\chi(supp(q)|E)=0$, then $wt_L(c_a)=|L|$ and let $f'_{26}$ be the frequency in this case.
%				\item If $\chi(supp(p)|D)=0$,  $\chi(supp(q)|F)=0$, $\chi(supp(q)|E)=1$, then $wt_L(c_a)=|L|$ and let $f'_{27}$ be the frequency in this case.
%				\item If $\chi(supp(p)|D)=0$,  $\chi(supp(q)|F)=1$, $\chi(supp(q)|E)=0$, then $wt_L(c_a)=|L|$ and let $f'_{28}$ be the frequency in this case.
%				\item If $\chi(supp(p)|D)=0$,  $\chi(supp(q)|F)=1$, $\chi(supp(q)|E)=1$, then $wt_L(c_a)=|L|$ and let $f'_{29}$ be the frequency in this case. Now, $\sum\limits_{i=26}^{29}f'_i=(2^m-2^{m-|D|})$.
%				\item If $\chi(supp(p)|D)=1$,  $\chi(supp(q)|F)=0$, $\chi(supp(q)|E)=0$, then $wt_L(c_a)=|L|$ and $f'_{30}=(2^{m-|D\cup E\cup F|})(2^{|E\cap F|-|D\cap E\cap F|})(2^{|E|-|E\cap (D\cup F)|}-1)(2^{|F|-|F\cap (D\cup E)|}-1)$.
%				\item If $\chi(supp(p)|D)=1$,  $\chi(supp(q)|F)=0$, $\chi(supp(q)|E)=1$, then $wt_L(c_a)=|L|$ and $f'_{31}=(2^{m-|D\cup E\cup F|})(2^{|F|-|F\cap(D\cup E)|}-1)$.
%				\item If $\chi(supp(p)|D)=1$,  $\chi(supp(q)|F)=1$, $\chi(supp(q)|E)=0$, then $wt_L(c_a)=|L|-\frac{1}{2}[(2^{|D|})(2^m-2^{|E|})(2^{|F|})]$ and $f'_{32}=(2^{m-|D\cup E\cup F|})(2^{|E|-|E\cap(D\cup F)|}-1)$.
%				\item If $\chi(supp(p)|D)=1$,  $\chi(supp(q)|F)=1$, $\chi(supp(q)|E)=1$, then $wt_L(c_a)=|L|-\frac{1}{2}[(2^{|D|})(2^m-2^{|E|})(2^{|F|})]+\frac{1}{2}[(2^{|D|})(2^{|E|})(2^{|F|})]$ and $f'_{33}=2^{m-|D\cup E\cup F|}-1$.
%			\end{itemize}
		%\end{itemize}
		\item [(8)] $p\ne0,q\ne0,r\ne0$, $wt_L(c_a)=|L|-\frac{1}{2}[(-2^{|D|}\chi(supp(p)|D))(-2^{|E|}\chi(supp(r)|E))(-2^{|F|}\\\chi(supp(q)|F))]-\frac{1}{2}[(-2^{|D|}\chi(supp(p)|D))(2^m\delta_{0,(q+r)}-2^{|E|}\chi(supp(q+r)|E))(-2^{|F|}\chi(supp(q)|F))]$.
		\begin{itemize} 
			\item Let $q\ne r$, $wt_L(c_a)=|L|-\frac{1}{2}[(-2^{|D|}\chi(supp(p)|D))(-2^{|E|}\chi(supp(r)|E))\\(-2^{|F|}\chi(supp(q)|F))]-\frac{1}{2}[(-2^{|D|}\chi(supp(p)|D))(-2^{|E|}\chi(supp(q+r)|E))(-2^{|F|}\chi(supp(q)|F))]$.
%			\begin{itemize}
%				\item $p\ne q,~p\ne r$
			\begin{itemize}
				\item If $\chi(supp(p)|D)=0$,  $\chi(supp(q)|F)=0$, $\chi(supp(r)|E)=0$,
				$\chi(supp(q+r)|E)=0$, then $wt_L(c_a)=|L|$ and let $f_{34}$ be the frequency in this case.
				\item If $\chi(supp(p)|D)=0$,  $\chi(supp(q)|F)=0$, $\chi(supp(r)|E)=0$,
				$\chi(supp(q+r)|E)=1$, then $wt_L(c_a)=|L|$ and let $f_{35}$ be the frequency in this case.
				\item If $\chi(supp(p)|D)=0$,  $\chi(supp(q)|F)=0$, $\chi(supp(r)|E)=1$,
				$\chi(supp(q+r)|E)=0$, then $wt_L(c_a)=|L|$ and let $f_{36}$ be the frequency in this case.
				\item If $\chi(supp(p)|D)=0$,  $\chi(supp(q)|F)=0$, $\chi(supp(r)|E)=1$,
				$\chi(supp(q+r)|E)=1$, then $wt_L(c_a)=|L|$ and let $f_{37}$ be the frequency in this case. %Now, $\sum\limits_{i=34}^{37}f_i=(2^m-2^{m-|D|})[(2^m-2^{m-|F|})(2^m-2)]$.
				\item If $\chi(supp(p)|D)=0$,  $\chi(supp(q)|F)=1$, $\chi(supp(r)|E)=0$,
				$\chi(supp(q+r)|E)=0$, then $wt_L(c_a)=|L|$ and let $f_{38}$ be the frequency in this case.
				\item If $\chi(supp(p)|D)=0$,  $\chi(supp(q)|F)=1$, $\chi(supp(r)|E)=0$,
				$\chi(supp(q+r)|E)=1$, then $wt_L(c_a)=|L|$ and let $f_{39}$ be the frequency in this case.
				\item If $\chi(supp(p)|D)=0$,  $\chi(supp(q)|F)=1$, $\chi(supp(r)|E)=1$,
				$\chi(supp(q+r)|E)=0$, then $wt_L(c_a)=|L|$ and let $f_{40}$ be the frequency in this case.
				\item If $\chi(supp(p)|D)=0$,  $\chi(supp(q)|F)=1$, $\chi(supp(r)|E)=1$,
				$\chi(supp(q+r)|E)=1$, then $wt_L(c_a)=|L|$ and let $f_{41}$ be the frequency in this case. %Now, $\sum\limits_{i=38}^{41}f_i=(2^m-2^{m-|D|})[(2^{m-|F|}-1)(2^m-2)]$. 
				Now, $\sum\limits_{i=34}^{41}f_i=(2^m-2^{m-|D|})(2^m-1)(2^m-2)$.
				\item If $\chi(supp(p)|D)=1$,  $\chi(supp(q)|F)=0$, $\chi(supp(r)|E)=0$,
				$\chi(supp(q+r)|E)=0$, then $wt_L(c_a)=|L|$ and let $f_{42}$ be the frequency in this case.
				\item If $\chi(supp(p)|D)=1$,  $\chi(supp(q)|F)=0$, $\chi(supp(r)|E)=0$,
				$\chi(supp(q+r)|E)=1$, then $wt_L(c_a)=|L|$ and let $f_{43}$ be the frequency in this case.
				\item If $\chi(supp(p)|D)=1$,  $\chi(supp(q)|F)=0$, $\chi(supp(r)|E)=1$,
				$\chi(supp(q+r)|E)=0$, then $wt_L(c_a)=|L|$ and let $f_{44}$ be the frequency in this case.
				\item If $\chi(supp(p)|D)=1$,  $\chi(supp(q)|F)=0$, $\chi(supp(r)|E)=1$,
				$\chi(supp(q+r)|E)=1$, then $wt_L(c_a)=|L|$ and let $f_{45}$ be the frequency in this case. Now, $\sum\limits_{i=42}^{45}f_i=(2^{m-|D|}-1)[(2^m-2^{m-|F|})(2^m-2)]$.
				\item If $\chi(supp(p)|D)=1$,  $\chi(supp(q)|F)=1$, $\chi(supp(r)|E)=0$,
				$\chi(supp(q+r)|E)=0$, then $wt_L(c_a)=|L|$ and $f_{46}=(2^{m-|D|}-1)[2^m(2^{m-|F|}-1)+2^{m-|E|}(1+2^{m-|E\cup F|}-2^{m+1-|F|})]$.
				\item If $\chi(supp(p)|D)=1$,  $\chi(supp(q)|F)=1$, $\chi(supp(r)|E)=0$,
				$\chi(supp(q+r)|E)=1$, then $wt_L(c_a)=|L|+\frac{1}{2}[(2^{|D|})(2^{|E|})(2^{|F|})]$ and let $f_{47}$ be the frequency in this case.
				\item If $\chi(supp(p)|D)=1$,  $\chi(supp(q)|F)=1$, $\chi(supp(r)|E)=1$,
				$\chi(supp(q+r)|E)=0$, then $wt_L(c_a)=|L|+\frac{1}{2}[(2^{|D|})(2^{|E|})(2^{|F|})]$ and let $f_{48}$ be the frequency in this case. Now $f_{47}+f_{48}=(2^{m-|D|}-1)[2(2^{m-|E|}-1)(2^{m-|F|}-2^{m-|E\cup F|})]$.
				\item If $\chi(supp(p)|D)=1$,  $\chi(supp(q)|F)=1$, $\chi(supp(r)|E)=1$,
				$\chi(supp(q+r)|E)=1$, then $wt_L(c_a)=|L|+[(2^{|D|})(2^{|E|})(2^{|F|})]$ and $f_{49}=(2^{m-|D|}-1)[(2^{m-|E|}-2)(2^{m-|E\cup F|}-1)]$.
			\end{itemize}
			\item Let $q=r$, $wt_L(c_a)=|L|-\frac{1}{2}[(-2^{|D|}\chi(supp(p)|D))(-2^{|E|}\chi(supp(r)|E))\\(-2^{|F|}\chi(supp(q)|F))]-\frac{1}{2}[(-2^{|D|}\chi(supp(p)|D))(2^m-2^{|E|})(-2^{|F|}\chi(supp(q)|F))]$.
%			\begin{itemize}
%				\item $p\ne q$
			\begin{itemize}
				\item If $\chi(supp(p)|D)=0$,  $\chi(supp(q)|F)=0$, $\chi(supp(r)|E)=0$, then $wt_L(c_a)=|L|$ and let $f_{50}$ be the frequency in this case.
				\item If $\chi(supp(p)|D)=0$,  $\chi(supp(q)|F)=0$, $\chi(supp(r)|E)=1$, then $wt_L(c_a)=|L|$ and let $f_{51}$ be the frequency in this case.
				\item If $\chi(supp(p)|D)=0$,  $\chi(supp(q)|F)=1$, $\chi(supp(r)|E)=0$, then $wt_L(c_a)=|L|$ and let $f_{52}$ be the frequency in this case.
				\item If $\chi(supp(p)|D)=0$,  $\chi(supp(q)|F)=1$, $\chi(supp(r)|E)=1$, then $wt_L(c_a)=|L|$ and let $f_{53}$ be the frequency in this case. Now, $\sum\limits_{i=50}^{53}f_i=(2^m-2^{m-|D|})(2^m-1)$.
				\item If $\chi(supp(p)|D)=1$,  $\chi(supp(q)|F)=0$, $\chi(supp(r)|E)=0$, then $wt_L(c_a)=|L|$ and $f_{54}=(2^{m-|D|}-1)[2^m-(2^{m-|E\cup F|})\{(2^{|E|-|E\cap F|})+(2^{|F|-|E\cap F|})-1\}]$.
				\item If $\chi(supp(p)|D)=1$,  $\chi(supp(q)|F)=0$, $\chi(supp(r)|E)=1$, then $wt_L(c_a)=|L|$ and $f_{55}=(2^{m-|D|}-1)[(2^{m-|E\cup F|})(2^{|F|-|E\cap F|}-1)]$.
				\item If $\chi(supp(p)|D)=1$,  $\chi(supp(q)|F)=1$, $\chi(supp(r)|E)=0$, then $wt_L(c_a)=|L|-\frac{1}{2}[(2^{|D|})(2^m-2^{|E|})(2^{|F|})]$ and $f_{56}=(2^{m-|D|}-1)[(2^{m-|E\cup F|})(2^{|E|-|E\cap F|}-1)]$.
				\item If $\chi(supp(p)|D)=1$,  $\chi(supp(q)|F)=1$, $\chi(supp(r)|E)=1$, then $wt_L(c_a)=|L|+\frac{1}{2}[(2^{|D|})(2^{|E|})(2^{|F|})]-\frac{1}{2}[(2^{|D|})(2^m-2^{|E|})(2^{|F|})]$ and $f_{57}=(2^{m-|D|}-1)[2^{m-|E\cup F|}-1]$.
			\end{itemize}
%			\item $p=q$
%			\begin{itemize}
%				\item If $\chi(supp(p)|D)=0$,  $\chi(supp(q)|F)=0$, $\chi(supp(r)|E)=0$, then $wt_L(c_a)=|L|$ and let $f'_{50}$ be the frequency in this case.
%				\item If $\chi(supp(p)|D)=0$,  $\chi(supp(q)|F)=0$, $\chi(supp(r)|E)=1$, then $wt_L(c_a)=|L|$ and let $f'_{51}$ be the frequency in this case.
%				\item If $\chi(supp(p)|D)=0$,  $\chi(supp(q)|F)=1$, $\chi(supp(r)|E)=0$, then $wt_L(c_a)=|L|$ and let $f'_{52}$ be the frequency in this case.
%				\item If $\chi(supp(p)|D)=0$,  $\chi(supp(q)|F)=1$, $\chi(supp(r)|E)=1$, then $wt_L(c_a)=|L|$ and let $f'_{53}$ be the frequency in this case. Now, $\sum\limits_{i=50}^{53}f'_i=(2^m-2^{m-|D|})$.
%				\item If $\chi(supp(p)|D)=1$,  $\chi(supp(q)|F)=0$, $\chi(supp(r)|E)=0$, then $wt_L(c_a)=|L|$ and $f'_{54}=(2^{m-|D\cup E\cup F|})(2^{|E\cap F|-|D\cap E\cap F|})(2^{|E|-|E\cap (D\cup F)|}-1)(2^{|F|-|F\cap (D\cup E)|}-1)$.
%				\item If $\chi(supp(p)|D)=1$,  $\chi(supp(q)|F)=0$, $\chi(supp(r)|E)=1$, then $wt_L(c_a)=|L|$ and $f'_{55}=(2^{m-|D\cup E\cup F|})(2^{|F|-|F\cap(D\cup E)|}-1)$.
%				\item If $\chi(supp(p)|D)=1$,  $\chi(supp(q)|F)=1$, $\chi(supp(r)|E)=0$, then $wt_L(c_a)=|L|-\frac{1}{2}[(2^{|D|})(2^m-2^{|E|})(2^{|F|})]$ and $f'_{56}=(2^{m-|D\cup E\cup F|})(2^{|E|-|E\cap(D\cup F)|}-1)$.
%				\item If $\chi(supp(p)|D)=1$,  $\chi(supp(q)|F)=1$, $\chi(supp(r)|E)=1$, then $wt_L(c_a)=|L|+\frac{1}{2}[(2^{|D|})(2^{|E|})(2^{|F|})]-\frac{1}{2}[(2^{|D|})(2^m-2^{|E|})(2^{|F|})]$ and $f'_{57}=2^{m-|D\cup E\cup F|}-1$.
%			\end{itemize}
		%\end{itemize}
		\end{itemize}
	\end{itemize}

We must note that $|Ker(f)|=|\{c_a~|~wt_L(c_a)=0\}|$. So, to compute the size of $C_L$, first we observe from Table \ref{Tab 3.1} that $wt_L(c_a)=0$, i.e., $((a\cdot l))_{l\in L}=0$ in the following three cases depending on the choices of $|D|,~|E|,~|F|$.
\begin{itemize}
	\item $a=0$ for all $|D|,~|E|,~|F|$. In this case $f_0=1$.
	\item $a=(0,q+ur)$ satisfying $\chi(supp(q)|F)=1$, $\chi(supp(r)|E)=1$, $\chi(supp(q+r)|E)=1$ for $|E|=|F|=m-1$. In this case $f_{15}=0$.
	\item $a=(p,ur)$ satisfying $\chi(supp(p)|D)=1$,  $\chi(supp(r)|E)=1$ for $|D|=|E|=m-1$. In this case $f_{25}=1$.
\end{itemize}
So, we have 
	\begin{equation*}
		|Ker(f)|=\begin{cases}
		2&if~|D|=|E|=m-1,\\
		1&otherwise.
		\end{cases}
	\end{equation*}
This gives the size of $C_L$.
\end{proof}

%\begin{corollary}
%	Let $|D|=|E|=|F|=m-1$. Then $C_L$ has parameters $[2^{3m-3},3m-1,2^{3m-4}]$.
%\end{corollary}
%
%\begin{corollary}
%	Let $|D|=|E|=|F|=n<m-1$. Then $C_L$ has parameters $[(2^m-2^n)^3,3m,(2^m-2^n)(2^m)(2^m-2^{n+1})]$.
%\end{corollary}

\begin{table}[h]
	\begin{center}
		\caption{Lee weight distribution of $C_L$} \label{Tab 3.1}
		\begin{tabular}{|c|c|c|}
			\hline
			Lee weight&frequency\\\hline
			$0$&$f_1=1$\\\hline
			$|L|$&$\sum\limits_{i=2}^{i=57}f_i,~i\ne3,6,7,13,14,15,18,19,$\\
			&$21,25,32,33,47,48,49,56,57$\\\hline
			$|L|+(2^m-2^{|D|})(2^{|E|})(2^m-2^{|F|})$&$f_3$\\\hline
			$|L|+\frac{1}{2}[(2^m-2^{|D|})(2^m-2^{|E|})(2^{|F|})]$&$f_6+f_{18}$\\\hline
			$|L|+\frac{1}{2}[(2^m-2^{|D|})(2^m-2^{|E|})(2^{|F|})-(2^m-2^{|D|})(2^{|E|})(2^{|F|})]$&$f_7+f_{19}$\\\hline
			$|L|-\frac{1}{2}[(2^m-2^{|D|})(2^{|E|})(2^{|F|})]$&$f_{13}+f_{14}$\\\hline
			$|L|-(2^m-2^{|D|})(2^{|E|})(2^{|F|})$&$f_{15}$\\\hline
			$|L|+(2^{|D|})(2^m-2^{|E|})(2^m-2^{|F|})$&$f_{21}$\\\hline
			$|L|-(2^{|D|})(2^{|E|})(2^m-2^{|F|})$&$f_{25}$\\\hline
			$|L|-\frac{1}{2}[(2^{|D|})(2^m-2^{|E|})(2^{|F|})]$&$f_{32}+f_{56}$\\\hline
			$|L|-\frac{1}{2}[(2^{|D|})(2^m-2^{|E|})(2^{|F|})-(2^{|D|})(2^{|E|})(2^{|F|})]$&$f_{33}+f_{57}$\\\hline
			$|L|+\frac{1}{2}[(2^{|D|})(2^{|E|})(2^{|F|})]$&$f_{47}+f_{48}$\\\hline
			$|L|+(2^{|D|})(2^{|E|})(2^{|F|})$&$f_{49}$\\\hline
		\end{tabular}
	\end{center}
\end{table}

\begin{table}[h]
	\begin{center}
		\caption{Lee weight distribution of $C_L$ if $|D|=|E|=|F|=m-1$} \label{Tab 3.2}
		\begin{tabular}{|c|c|c|}
			\hline
			Lee weight&frequency\\\hline
			$0$&$\frac{1}{2}(f_1+f_{15}+f_{25})=1$\\\hline
			$|L|=2^{3m-3}$&$\frac{1}{2}\sum\limits_{i=2}^{i=57}f_i,~(i\ne3,6,13,14,15,18,$\\
			&$21,25,32,47,48,49,56)=2^{3m-1}-6$\\\hline
			$|L|+(2^m-2^{|D|})(2^{|E|})(2^m-2^{|F|})=2^{3m-2}$&$\frac{1}{2}(f_3+f_{21}+f_{49})=1$\\\hline
			$|L|+\frac{1}{2}[(2^m-2^{|D|})(2^m-2^{|E|})(2^{|F|})]=2^{3m-3}+2^{3m-4}$&$\frac{1}{2}(f_6+f_{18}+f_{47}+f_{48})=2$\\\hline
			$|L|-\frac{1}{2}[(2^m-2^{|D|})(2^{|E|})(2^{|F|})]=2^{3m-4}$&$\frac{1}{2}(f_{13}+f_{14}+f_{32}+f_{56})=2$\\\hline
		\end{tabular}
	\end{center}
\end{table}

\begin{remark}\label{rem 3.3}
	Let $|D|=|E|=m-1$, i.e., $|Ker(f)|=2$. Then for each $a\in\mathcal{R}^m$, there exists $b\in\mathcal{R}^m$ such that $((a\cdot l))_{l\in L}=((b\cdot l))_{l\in L}$ as this happens if and only if $(((a-b)\cdot l))_{l\in L}=0$ (as $C_L$ is linear), i.e., $(a-b)\in Ker(f)$ and $|Ker(f)|=2$. So, in this case the frequencies in Table \ref{Tab 3.1} corresponding to a particular Lee weight will be half of that. In Table \ref{Tab 3.2}, the Lee weight distribution of $C_L$ is given for $|D|=|E|=|F|=m-1$.
\end{remark}

\begin{example}
	Let  $m=3$ and $D=\{1,2\}$, $E=\{1,3\}$, $F=\{2,3\}$. Then the Lee weight enumerator of the code $C_L$ (as in Theorem \ref{thm 3.2} and Table \ref{Tab 3.2}), denoted by $lwe_{C_L}(X,Y)$, is given by $lwe_{C_L}(X,Y)=X^{128}+2X^{96}Y^{32}+250X^{64}Y^{64}+2X^{32}Y^{96}+Y^{128}$.
\end{example}

\begin{example}
	Let  $m=3$ and $D=\{1\}$, $E=\{2\}$, $F=\{3\}$. Then the Lee weight enumerator of the code $C_L$ (as in Theorem \ref{thm 3.2} and Table \ref{Tab 3.1}), denoted by $lwe_{C_L}(X,Y)$, is given by $lwe_{C_L}(X,Y)=X^{432}+11X^{240}Y^{192}+24X^{228}Y^{204}+6X^{224}Y^{208}+416X^{216}Y^{216}+36X^{212}Y^{220}+6X^{208}Y^{224}+2X^{192}Y^{240}+4X^{180}Y^{252}+6X^{144}Y^{288}$.
\end{example}

\begin{example}
	Let  $m=4$ and $D=\{1,2,3\}$, $E=\{1,2,4\}$, $F=\{1,3,4\}$. Then the Lee weight enumerator of the code $C_L$ (as in Theorem \ref{thm 3.2} and Table \ref{Tab 3.2}), denoted by $lwe_{C_L}(X,Y)$, is given by $lwe_{C_L}(X,Y)=X^{1024}+2X^{768}Y^{256}+2042X^{512}Y^{512}+2X^{256}Y^{768}+Y^{1024}$.
\end{example}

\begin{example}
Let  $m=4$ and $D=\{1,2\}$, $E=\{2,3\}$, $F=\{3,4\}$. Then the Lee weight enumerator of the code $C_L$ (as in Theorem \ref{thm 3.2} and Table \ref{Tab 3.1}), denoted by $lwe_{C_L}(X,Y)$, is given by $lwe_{C_L}(X,Y)=X^{3456}+11X^{1920}Y^{1536}+24X^{1824}Y^{1632}+6X^{1792}Y^{1664}+4000X^{1728}Y^{1728}+36X^{1696}Y^{1760}+6X^{1664}Y^{1792}+2X^{1536}Y^{1920}+4X^{1440}Y^{2016}+6X^{1152}Y^{2304}$.
\end{example}

\section{Gray images and Minimal codes}\label{section 4}

First of all we observe that: Since $\Phi$ is a linear map and $C_L$ is a linear code over $\Z_2[u]$, $\Phi(C_L)$ is linear over $\Z_2$. Next we will use some characterizations of binary linear codes regarding orthogonality and minimality.

\begin{proposition}
	Let $C_L$ be the linear code over $\Z_2[u]$ as in Theorem \ref{thm 3.2}. Then the binary linear code $\phi(C_L)$ is self-orthogonal with respect to the Euclidean inner product if $|D|,|E|,|F|\ne0$.
\end{proposition}

\begin{proof}
	Let $0<|D|,|E|,|F|\le m-1$. We know that if all the codewords of a binary linear codes is of weight $4k$ for some $k\in\N$, then the linear code is self-orthogonal (see (ii) of \cite[Theorem 1.4.8]{huffman2010book}). Now, $(2^m-2^{|D|})\ge2^{|D|}\ge2$ (as $m>1$ and $0<|D|\le m-1$). So, from Table \ref{Tab 3.1}, it is easy to observe that 4 divides all the Lee weights. Hence $\phi(C_L)$ is self-orthogonal.
\end{proof}

\begin{example}
	Let $|D|=|E|=|F|=m-1$. Then $\phi(C_L)$ is self-orthogonal with parameters $[2^{3m-2},3m-1,2^{3m-4}]$.
\end{example}

\begin{example}
	Let $|D|=|E|=|F|=n<m-1$. Then $\phi(C_L)$ has parameters $[2(2^m-2^n)^3,3m,(2^m-2^n)(2^m)(2^m-2^{n+1})]$. Moreover, if $n\ne0$, $\phi(C_L)$ is self orthogonal.
\end{example}

We call a linear code $C$ over $\F_p$ is minimal if all of its nonzero codewords are minimal. Further a codeword $c\in C$ is said to be minimal if $supp(c')\subseteq supp(c)$ if and only if $c'=ac$, where $a(\ne0)\in\F_p$. In general, it is hard to determine whether a linear code is minimal or not but the famous Ashikhmin-Barg condition (\cite{barg1998minimal}) is a sufficient condition for a linear code to be minimal and it is relatively easy to check.

\begin{lemma}(Ashikhmin-Barg condition for p=2) \cite{barg1998minimal}\label{lem 4.1}
	Let $C$ be a linear code over $\Z_2$. Denote $w_0$ as the minimum(nonzero) and $w_{\infty}$ as the maximum Hamming weights. If $\frac{w_0}{w_{\infty}}>\frac{1}{2}$, then $C$ is minimal.
\end{lemma}

\begin{theorem}\label{thm 4.2}
	Let $C_L$ be the linear code over $\Z_2[u]$ as in Theorem \ref{thm 3.2} with $|D|=|E|=|F|=n(<m)$. Then the code $\Phi(C_L)$ is minimal if $n\le m-2$.
\end{theorem}

\begin{proof}
	Suppose $|D|=|E|=|F|=n(<m)$. Then $|L|=(2^m-2^n)^3$ and $(2^m-2^n)\ge2^n,\forall m,n$. Now we divide this proof in two cases.\\
	\textit{Case I}: Let $n<m-1$. Then we observe that $(2^m-2^n)>2^n$. Now with these assumptions and using Table \ref{Tab 3.1}, we have $w_0=|L|-[(2^m-2^n)(2^n)^2]=(2^m-2^n)(2^m)(2^m-2^{n+1})$ and $w_\infty=|L|+[(2^m-2^n)^2(2^n)]=(2^m-2^n)^2(2^m)$. So, $\frac{w_0}{w_\infty}=\frac{2^m-2^{n+1}}{2^m-2^n}$ (as $m\ne0$ and $n\ne m$). Now to satisfy Lemma \ref{lem 4.1}, it must be $2^{m+1}-2^{n+2}>2^m-2^n$ i.e. $2^m>2^{n+2}-2^n$ and this will be true if and only if $n+2\le m$, which is also our assumption in this case.\\
	\textit{Case II}: Let $n=m-1$. Then we observe that $|L|=(2^n)^3$ and $(2^m-2^n)=2^n$. Now with these assumptions and using Table \ref{Tab 3.2}, we have $w_0=|L|-\frac{1}{2}(2^n)^3=\frac{1}{2}|L|$ and $w_\infty=|L|+(2^n)^3=2|L|$. So, $\frac{w_0}{w_\infty}=\frac{1}{4}$ (as $n\ne0$). Since Lemma \ref{lem 4.1} is only a sufficient condition, we cannot conclude anything in this case. Hence the theorem.
\end{proof}

\begin{example}
	Let $m=3$ in Theorem \ref{thm 4.2}. Then 
	\begin{itemize}
		\item [(1)] If $n=2$, we can not conclude anything about the minimality of the binary linear self-orthogonal four-weight code $\Phi(C_L)$ having parameters $[128,8,32]$.
		\item [(2)] If $n=1$, then $\Phi(C_L)$ is a minimal binary linear self-orthogonal nine-weight code with parameters $[432,9,192]$.
		\item [(3)] If $n=0$, then $\Phi(C_L)$ is a minimal binary linear seven-weight code with parameters $[686,9,336]$.
	\end{itemize}	
\end{example}

\begin{example}
Let $m=4$ in Theorem \ref{thm 4.2}. Then 
\begin{itemize}
	\item [(1)] If $n=3$, we can not conclude anything about the minimality of the binary linear self-orthogonal four-weight code $\Phi(C_L)$ having parameters $[1024,11,256]$.
	\item [(2)] If $n=2$, then $\Phi(C_L)$ is a minimal binary linear self-orthogonal nine-weight code with parameters $[3456,12,1536]$.
	\item [(3)] If $n=1$, then $\Phi(C_L)$ is a minimal binary linear self-orthogonal nine-weight code with parameters $[5488,12,2688]$.
	\item [(4)] If $n=0$, then $\Phi(C_L)$ is a minimal binary linear seven-weight code with parameters $[6750,12,3360]$.
\end{itemize}	
\end{example}

%We call a code optimal if it has the highest minimum distance for given length and dimension. Griesmer bound is an upper bound for linear codes and is a generalization of the Singleton bound. Thus, a linear code achieving the Griesmer bound with equality is optimal.
%
%\begin{theorem} (Griesmer bound) Let $C$ be a binary code with parameters $[n,k,d],~k\ge1$. Then $n\ge\sum\limits_{i=0}^{k-1}\lceil\frac{d}{2^i}\rceil$, where $\lceil n\rceil$ denotes the smallest integer greater than or equal to $n$.
%\end{theorem}

\section{Conclusion}\label{section5}
In \cite{wu2021mixedtrace}, the authors constructed few-Lee weight additive codes over the mixed alphabet $\Z_p\Z_p[u],u^2=0$ with optimal Gray images by using a suitable defining set. In \cite{wang2021mixeddown}, the authors used the mixed alphabet ring  $\Z_p\Z_p[u]$ with $u^2=u$ and $p>2$ to study some few-Lee weight additive codes by employing down-sets. In this paper, we use the mixed alphabet ring $\Z_2\Z_2[u]$ to construct a class of few-Lee weight linear codes using simplicial complexes generated by a single element and also have an infinite family of self-orthogonal binary minimal codes over as Gray images. This shows that like in the case of single alphabet rings, we may employ all three popular means of constructing few-Lee weight codes including simplicial complexes when mixed alphabet rings are involved and hence fully answer the question raised in \cite{wang2021mixeddown} about the use of simplicial complexes or down-sets in the case of mixed alphabet rings. 

%\section*{Declarations}
%\subsection*{Ethical Approval and Consent to participate} Not applicable.
%\subsection*{Consent for publication} Not applicable.
%\subsection*{Availability of supporting data} Not applicable.
%\subsection*{Competing interests} The authors declare that they have
%no competing interests.
%\subsection*{Funding} Not applicable.
%\subsection*{Authors' contributions} PKK and NKM have equal contributions.
%\subsection*{Acknowledgements} Not applicable.

\end{document}